\documentclass{llncs}
\usepackage{listings}
\usepackage{algorithm} 
\usepackage{algpseudocode} 
\usepackage{fancyhdr}
\usepackage{flushend}
\usepackage{amsfonts,amssymb,amsmath,alltt}
\usepackage{stmaryrd}
\usepackage{xspace}
\usepackage{graphicx}
\usepackage{color}

\usepackage{makeidx}
\usepackage[utf8]{inputenc}
\usepackage{url}
\usepackage[english]{babel}

\newenvironment{ttbox}{\begin{alltt}\ttbraces\small\tt}%
                      {\end{alltt}}
\def\ttbraces{\let\.=\nobreak\chardef\{=`\{\chardef\}=`\}\chardef\|=`\\}

\newcommand\ttand{\mbox{{$\land$}}}
\newcommand\ttor{\mbox{{$\lor$}}}

\newcommand\ttcup{\mbox{{$\cup$}}}
\newcommand\ttfun{\mbox{{$\Rightarrow$}}}

\newcommand\ttimp{\mbox{{$\longrightarrow$}}}
\newcommand\ttequiv{\mbox{{$\equiv$}}}
\newcommand\ttexists{\mbox{{$\exists$}}}
\newcommand\ttforall{\mbox{{$\forall$}}}
\newcommand\ttneg{\mbox{{$\neg$}}}
\newcommand\ttneq{\mbox{{$\neq$}}}
\newcommand\ttin{\mbox{{$\in$}}}

\newcommand\ttImp{\mbox{{$\Longrightarrow$}}}
\newcommand\ttlam{\mbox{\( \lambda \)}}
\newcommand\tttimes{\mbox{\( \times \)}}

\newcommand\ttto[1]{\mbox{{$\to^{#1}$}}}
\newcommand\ttIFC[1]{\mbox{{\text{IFC}$_{#1}$}}}

\newcommand\ttrelI{\mbox{{$\to_{i}$}}}

\newcommand\ttrel[1]{\mbox{{$\to_{#1}$}}}
\newcommand\ttrelstar[1]{\mbox{{$\to_{#1}^*$}}}
\newcommand\ttalpha{\mbox{{$\alpha$}}}
\newcommand\ttbeta{\mbox{{$\beta$}}}

\newcommand\tttau{\mbox{{$\tau$}}}
\newcommand\ttsubseteq{\mbox{{$\subseteq$}}}

\newcommand\ttsone{\mbox{{\texttt{s}$_1$}}}

\newcommand\ttvdash{\mbox{{$\vdash$}}}
\newcommand\ttref[1]{\mbox{\(\sqsubseteq_{#1}\)}}

\newcommand\ttNI[1]{\mbox{{$\text{NI}_{#1}$}}}

\newcommand\ttsigma{\mbox{{$\sigma$}}}
\newcommand\ttmref[1]{\mbox{{$\sqsubseteq_{#1}$}}}
\newcommand\ttmeref{\ttmref{\mathcal{E}}}

\newcommand\ttecal{\mbox{$\mathcal{E}$}}
\newcommand\ttimg{\mbox{\texttt{`}}}
\newcommand\ttmaplet{\mbox{\(\to\)}}
\newcommand{\ttgzero}{\mbox{$\texttt{g}_0$}}
\newcommand{\ttgone}{\mbox{$\texttt{g}_1$}}

\newcommand{\tttleqA}[1]{\mbox{$\trianglelefteq_{#1}$}}
\newcommand{\tttleA}[1]{\mbox{$\triangleleft_{#1}$}}

\newcommand\ttszero{\mbox{{$s_0$}}}
\newcommand\ttxzero{\mbox{{$x_0$}}}

\newcommand\ttany[1]{\mbox{{${#1}$}}}

\newcommand\ttanyany[2]{\mbox{{${#1}_{#2}$}}}
\newcommand\ttsim[1]{\mbox{{$\sim_{#1}$}}}
\newcommand\ttrestr[2]{\mbox{$\left.{#1}\right|_{#2}$}}
\newcommand\ttshadow[1]{\mbox{$\texttt{shadow}_{\texttt{#1}}$}}

\begin{document}
\frontmatter
  
\mainmatter
\title{Shadowing the Isabelle Insider and Infrastructure Framework}
  \author{Florian Kamm\"uller}
\institute{Middlesex University London and TU Berlin\\
\email{florian.kammuller@gmail.com}
}
\maketitle
\begin{abstract}
In this paper, we extend the process of Security Engineering for the Isabelle Insider and Infrastructure framework (IIIf) by introducing Information Flow Security (IFC). To formalize the absence of information flows to lower levels, we use a concept of a ``Shadow'' inspired by Morgan. We relate it to the classical notion of Noninterference (NI) formalised in the IIIf. Apart from being an elegant concept, Morgan's concept of a shadow is interesting because it addresses a phenomenon called the ``refinement paradox'':
information flow security is known to be not preserved by specification refinements in general. We use the formalisation of shadow and its equivalence to NI to exhibit conditions for a secure refinement for IIIf. As a running example to illustrate the problem, the concepts and the solution, we use an example of a flightradar system specification.
\end{abstract}  
\section{Introduction}
The Isabelle Insider and Infrastructure (IIIf) \cite{kam:20a} provides formal specifications of
infrastructures with actors and policies within the interactive proof assistant Isabelle. It therefore
supports formalising and verifying security and privacy of systems centered around human actors \cite{kam:26iiif}.
The IIIf as a framework provides Isabelle formalisations of attack trees, Kripke structures and model checking
as a basis for expressing infrastructure systems as state transition systems. A notion of trace based refinement
is formalised including property preservation.
However, intricate notions like Information Flow Control (IFC) are known to be not within the reach of model
checking verification \cite{mcl:94}. Noninterference (NI) expresses the preservation of an indistinguishability relation
over states. It is thus a property over sets of system executions rather than a property on single system executions.
NI also does not refine: a specification could have the NI property, but a refinement -- even
using a rigorous notion of property preserving refinement -- might fail to preserve NI.
This paradox is coined ``the refinement paradox'' \cite{mor:09}.
Since the IIIf aims at security engineering of system specification expressed as the Refinement-Risk cycle (RR-cycle)
\cite{kam:20a,kam:23b} this is a flaw: it neither preserves NI.
In this paper, we address this issue and develop a notion of shadow following the concepts and outline of \cite{mor:09}
to characterize security preserving refinements.

As a preparation, we first present the integration of NI into the IIIf motivating it with a concise example of a
privacy-critical flightradar system introduced in Section \ref{sec:flightradar}.
We re-use here this case-study \cite{kam:24a} because it serves well for illustrating physical aspects of systems
and also suits the demonstration of the concepts of IFC and refinement preservation.
This 
formalisation of the flightradar system in IIIf 
introduces access control labels into the model (Section \ref{sec:ifcmove}).
In Section \ref{sec:ni}, we introduce the notion of NI and its proof on the example.

The novel contribution of this paper is an additional layer of 
IIIf's security engineering:
a methodology for defining a shadow, showing its equivalence to NI and investigating its refinement in IIIf.
The application to the flightradar system example serves to illustrate implicit information flows and NI in IIIf
proved on the example (Sections \ref{sec:flightradar} and \ref{sec:ni}). We then provide a definition of shadow
(Section \ref{sec:shaiiif}) and illustrate how it is applied to the flightradar system leading to identifying an
information flow invariant \texttt{\ttIFC{\texttt{a}}}.
We show that \texttt{\ttIFC{\texttt{a}}} is equivalent to NI.
Finally, we investigate the relationship between shadow, NI and refinement (Section \ref{sec:secref}) before we
conclude and discuss related work in Section \ref{sec:concl}.
The source code of the IIIf with IFC and its application to the flightradar system is available
online \cite{kam:26iiif}.
To set the scene, we use the remainder of this section to introduce IIIf.

\subsection{Isabelle Insider and Infrastructure framework IIIf}
The IIIf is an extension of Higher Order Logic (HOL) in the interactive generic theorem
prover Isabelle/HOL \cite{npw:02}. It thus augments the automated reasoning capabilities of Isabelle
with dedicated support for modeling and proving of systems with physical and logical components,
actors and policies. The IIIf has been designed for the analysis of insider threats. However, the
implemented theory of temporal logic combined with Kripke structures and its generic notion of state
transitions are a perfect match to be combined with  attack trees into a process for formal security
engineering \cite{suc:16} including an accompanying framework \cite{kam:19a}.
A number of case studies (see Section \ref{sec:concl})
have contributed to shape the IIIf into a general framework for
the state-based security analysis of infrastructures with policies and actors. Temporal logic
and Kripke structures are deeply embedded into Isabelle/HOL thereby
enabling meta-theoretical proofs about the foundations: for example, equivalence between attack trees 
and CTL statements have been established \cite{kam:18b} providing sound foundations for applications.
This foundation provides a generic notion of state transition on which attack trees and
temporal logic can be used to express properties for applications.

The use of Isabelle/HOL 
for the construction of the IIIf permit a meta-theoretical approach: 
concepts like temporal logics, model checking, and now IFC can be formalized in the logic while at
the same time these concepts can be applied to case studies.
However, there are limits to abstraction due to types. The equivalence between NI and shadow
can only be formalized and proved for infrastructures in general on paper because it quantifies over
differently typed components of an infrastructure state. 
The methodology lives at a level even too high for Higher Order Logic. 

\section{Flightradar}
\label{sec:flightradar}
Privacy as well as security may be endangered when there appear daily differences on airplane routes.
Such irregularities may reveal secret information, for example, certain areas cannot be overflown because
of security alerts, or security or privacy critical passengers are on board of the airplanes so that higher
security precautions impose maximally safe airplane routing. 
These precautions are intended to increase security and privacy by physical measures. However, on the level
of confidential information they represent implicit information flows.
Geographic Information Systems (GIS) offer the possibility to obfuscate private information by blurring
the details of the maps thus establishing some level of privacy \cite{kam:24a}.  
However, this does not scale to dynamic information as needed for example for airtraffic control systems
where blurring could cause implicit information flows.
In formal techniques for software verification such implicit information flows are well known and
countermeasures are well understood in an area called Information Flow Control (IFC).
In addition to Access Control, which regulates access to objects, IFC ``specifies valid channels along
which information may flow'' \cite{den:82}.
It serves to secure flows of information in software systems by labeling values with security classes
and enforcing the security policy by only allowing information to flow between classes if they are in a
flow relation \cite{dd:77}. IFC has become a standard technique in the security specification and verification
of software systems of all shades but is less well understood when it comes to more heterogeneous systems
including physical as well as logical aspects and human actors.

\subsection{Air-traffic information systems}
The American Federal Aviation Administration (FAA) publishes radar data but systems like flightradar24
(see Figure \ref{fig:flightradar}) rely on planes knowing their own position quite accurately via GPS
and broadcasting them over the Automatic Dependent Surveillance Broadcast (ADS-B). Partially private
ADS-B spotters in various locations receive this data and forward it to the web-service. Not all planes can
be spotted, for example military planes and private planes. It depends on them having ADS-B systems on board.
Also, it is possible to agree with the web services not to show one's plane -- or to show it in a more
secure way. 
The IIIf \cite{kam:24a} provides suitable foundations and formalizes the specification of this system.
We will use this as a basis in this paper.
\begin{figure*}
\vspace{-.5cm}
  \begin{center}
  \includegraphics[scale=.3]{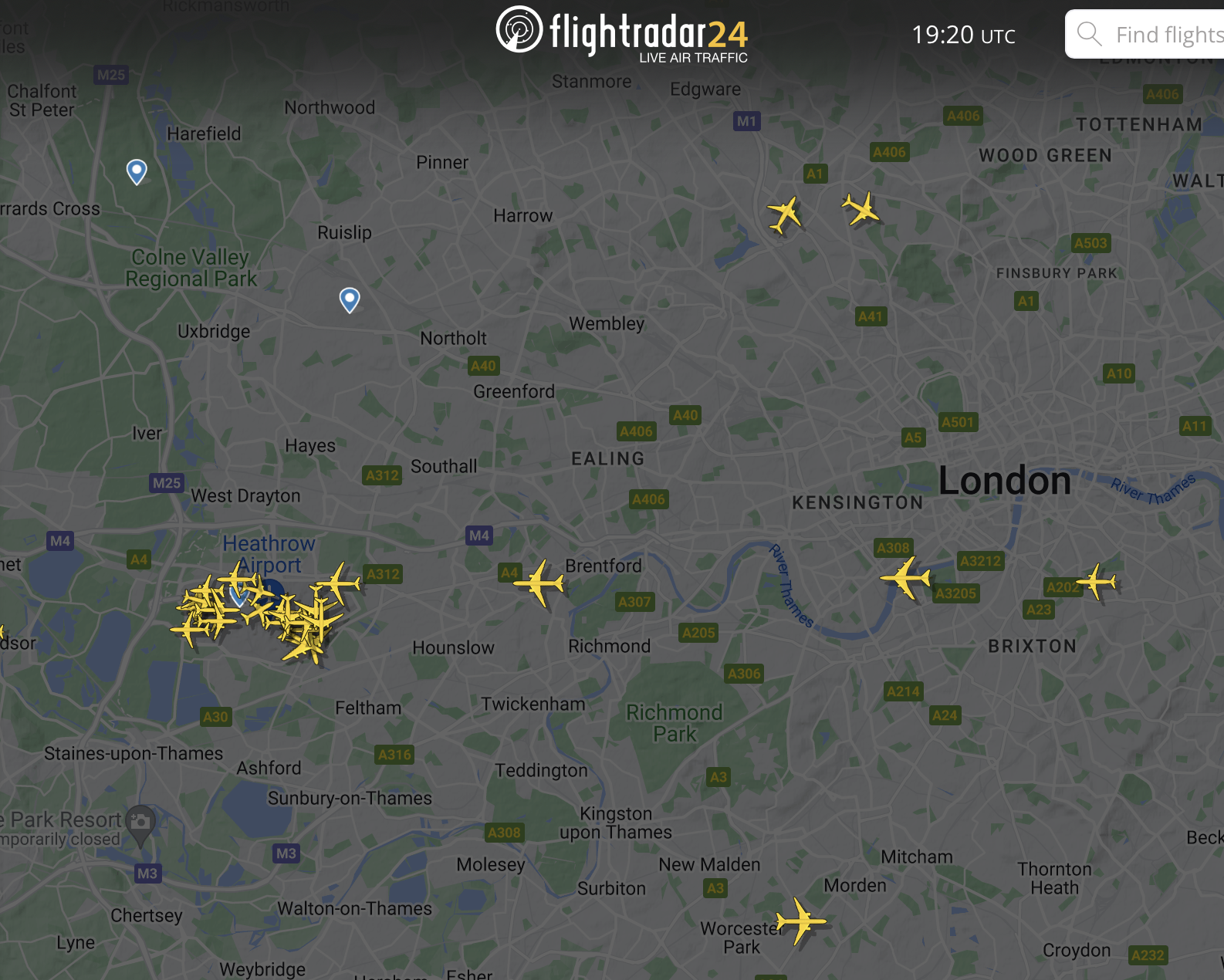}
\end{center}
\vspace{-.5cm}
\caption{Flightradars like  flightradar24 \cite{fli:24} visualize airtraffic routes publicly.}
\label{fig:flightradar}
\end{figure*}
In the IIIf, infrastructures are supposed to be presented as graphs of locations (which could be physical or logical)
and actors as well as local policies are attached to these locations. There is no fixed location type; in the
flightradar case study, locations are represented as pairs of natural numbers thus representing a discrete 2-dimensional
grid.
Infrastructures entail the graphs and the policies and constitute the {\it states} of the system.
The underlying theory of Kripke structures and temporal logics CTL is formalised as the core framework of the IIIf
over a generic state transition relation (using Isabelle's type classes) and building up the entire CTL logic from
fixpoint definitions and Tarski's fixpoint theory.
This generic core framework is then instantiated with an application specific graph type that extends the above mentioned
components by adding some more application specific ones and defining a specific state transition relation for the
application scenario. This has been done for the flightradar system in \cite{kam:24a}. It will be used here
to provide the notion of NI for infrastructures that we want to investigate and extend with respect to refinement.

\subsection{Security considerations}
The function \texttt{circumvent} computes an exceptional flightpath based on the current position \texttt{(l,l')} and the
planned next one \texttt{(n, n')} 
on the normal route to compute the closest alternative position.

This circumvention now would represent an illicit information flow, that is, unauthorized users learn
confidential information about the criticality of the location.

As is common and meaningful in security models for program security, it is assumed that Eve knows the implementation
of the flight radar 
system as well as all the visible output data, like the positions of the planes in \texttt{planes} but she does
not know the security critical data, that is \texttt{critloc}.
However, via the above implicit flow Eve effectively gets to know \texttt{critloc} from the behaviour of the
airplanes:

\noindent%
{\it a plane circumvents a position} \texttt{(n,n')} {\bf if and only if} \texttt{critloc (n, n')}.

\noindent%
Therefore, there is a leakage of confidential critical information of the positions to a visible output. In
order to specify and potentially provably exclude this security issue, it is necessary to first add an explicit way
to express security. This is done next by introducing security labels to the infrastructure specification.

Every value (data or identity) that is used or communicated within the evolution of the infrastructure,
that is, in the execution of the state transition has a security label associated to it representing the
access control for that item; for an object, we refer to this as \textit{level} (\texttt{lv}), for a subject
we refer to it as \textit{clearance} (\texttt{clrnc}).
There is a module in the IIIf for  Information Flow Control \texttt{IFC.thy} \cite{kam:24a}, that formalises
a generic lattice model for security classes for access control and in particular provides instances
of this lattice for multi-level security (MLS) \cite{de:76} and the Decentralized Label Model (DLM) \cite{ml:98}.
However, it may well be also instantiated to other access control models, for example Role-Based Access Control (RBAC) or Attribute Based Access Control (ABAC).

\subsection{State transition with security labels and hiding}
\label{sec:ifcmove}

The infrastructure graph is defined as a datastructure \texttt{igraph} whose components \texttt{gra},
\texttt{planes}, \texttt{routes}, and \texttt{critloc} represent
the parts as:
a set of pairs of locations -- the coordinates of the map for the flightradar system;
the identities of the airplanes at each location;
the routes of each airplane and the criticality of each location. The latter is a boolean valued function
(predicate) as criticality is assumed to be a boolean flag.
The constructor \texttt{Lgraph} puts these components into an \texttt{igraph}.

Each component of the infrastructure graph datatype \texttt{igraph}
is paired with an access control label of type \texttt{ac} as first element. 
In addition, the (secret) circumvented position of a plane is added as a new component
\texttt{critpos}.
This will allow us to implement a different behaviour of the flight radar that avoids revealing information by
implicit flows. 
The updated \texttt{igraph} datatype looks now as follows.
\begin{ttbox}
{\bf{datatype}} igraph = Lgraph
{\bf gra:} ac \tttimes (location \tttimes location)set 
{\bf planes:} ac \tttimes ((location \tttimes location) \ttfun identity set)
{\bf routes:} ac \tttimes (identity \ttfun (location \tttimes location)list)
{\bf critloc:} ac \tttimes ((location \tttimes location) \ttfun bool)
{\bf critpos:} ac \tttimes (identity \ttfun (location \tttimes {location})option)
\end{ttbox}
The \texttt{ac} labels are the first component of each \texttt{igraph} element and can thus be accessed using
the function \texttt{fst} that returns the first element of a pair. For convenience, we define the following
syntactic sugar that allows henceforth to write \texttt{lv x} to access the \texttt{ac} label of any of the fields
of an \texttt{igraph}. Similarly, we provide a selector for the \textit{value} by the pair projection \texttt{snd} 
\begin{ttbox}
{\bf definition} lv :: (\ttalpha \tttimes \ttbeta) \ttfun \ttalpha {\bf where} lv = fst
{\bf definition} vl :: (\ttalpha \tttimes \ttbeta) \ttfun \ttalpha {\bf where} vl = snd

\end{ttbox}
The security labels are assigning an overall label to each of the components. For example for \texttt{critloc},
this label comprises all the critical and non-critical locations. That is, the labels do not specify on a fine
grained detail the security level of each location.

To integrate the access control given by the security labels the flight radar specification 
uses hiding to cover the real position of an airplane if it is circumvented around a critical location.
Consequently, in the two semantic rules for moving an airplane, a check is added to evaluate whether
the next position of the airplane is critical.
There are two rules for the move action: \texttt{move} and \texttt{move\_crit}.
The first rule is for the case \texttt{\ttneg critloc G (n, n')} saying that the next position on the
flightpath is not critical.
Here, the rule simply updates the \texttt{planes} component at the current position \texttt{(l,l')} and
the next position \texttt{(n,n')} on the flightpath by removing the airplane \texttt{f} from the
former and placing it on the latter position. 
The action also sets back the \texttt{critpos G} flag for the flight to delete a previous
``circumvented'' position in case the previous position had been critical (see the following action \texttt{move\_crit}
for more detail on the latter point).
\begin{ttbox}
{\bf{move}}: G = graphI I \ttImp (l, l') \ttin  gra G \ttImp f \ttin planes G (l, l') \ttImp
 (n,n') = hd(routes G f) \ttImp \ttneg critloc G (n, n') \ttImp
 I' = Infrastructure
 (Lgraph
 (gra G)
 (lv(planes G), vl(planes G)((l,l'):= vl(planes G)(l,l')-\{f\})
                            ((n,n'):= vl(planes G))(n,n') \ttcup \{f\})
 (lv(routes G), vl(routes G)(f := tl (vl(routes G) f)))
 (critloc G)
 (lv (critpos G), vl (critpos G)(f := None)))
 \ttImp I \ttrelI I' 
\end{ttbox}
In the case for moving over a critical location, 
the plane is not circumvented in the component \texttt{planes} but apparently remains ``on course''.
Instead, the additional component \texttt{critpos} registers that the flight is in fact circumvented.
Modeling the circumvention this way allows hiding the actual position easily by using the \texttt{planes}
component as the ``public'' (or official) position while hiding the real position in \texttt{critpos}.
\begin{ttbox}
{\bf{move\_crit}}: G = graphI I \ttImp (l, l') \ttin gra G \ttImp f \ttin planes G (l, l') \ttImp
 (n,n') = hd(routes G f) \ttImp critloc G (n, n') \ttImp
 I' = Infrastructure
 (Lgraph
 (gra G)
 (lv(planes G), (vl(planes G))((l,l'):= (vl(planes G))(l,l')-\{f\})
                              ((n,n'):= (vl(planes G))(n,n')\ttcup\{f\}))
 (lv(routes G), (vl(routes G)(f := tl (vl(routes G) f)))
 (critloc G) 
 (lv (critpos G), (vl(critpos G))(f \ttmaplet (circumvent (l,l')(n,n'))))))
\ttImp I \ttrelI I' 
\end{ttbox}
Those airplanes that are circumventing critical positions are registered in the \texttt{planes}
component of the infrastructure graph as being {\it at} the critical location. This implementation hides
that these planes are actually circumvented; the actual position is only recorded in the state
component \texttt{critpos}.

\subsection{Noninterference}
\label{sec:ni}
IFC can be expressed for programming languages as the notion of noninterference (NI) \cite{gm:82}.
In our framework of state-based system specification, NI is expressed as an equivalence over sequences
of state transitions. For other formalisms, that are more event oriented, like reactive systems, it may
be expressed similarly but over event traces.
Intuitively, NI characterizes that various executions of a system do not reveal any secret information
to a specific viewpoint -- that of the attacker. The attacker's viewpoint can be identified 
by an access control class, generally that of an observer of the system transitions.
From the viewpoint of this  observer, that is, its access control class, the system states appear equal
even if they contain different secret information. 
This equivalence relation is called indistinguishability.
Formally, the indistinguishability relation is first defined as a relation over infrastructure graphs 
by the following case statement
where $A$ is the set of actors of the infrastructure and the security class of the observer (attacker)
is \texttt{a}.
\begin{ttbox}
{\bf{definition}} indistinguishability\_g  
{\bf{where}} indistinguishability\_g \ttgzero a \ttgone \ttequiv
if lv (gra \ttgzero) \tttleqA{A} a then vl(gra \ttgzero) = vl(gra \ttgone) else True \ttand
if lv (planes \ttgzero) \tttleqA{A} a then vl(planes \ttgzero) = vl(planes \ttgone) else True \ttand 
if lv (routes \ttgzero) \tttleqA{A} a then vl(routes \ttgzero) = vl(routes \ttgone) else True \ttand
if lv (critloc \ttgzero) \tttleqA{A} a then vl(critloc \ttgzero) = vl(critloc \ttgone) else True \ttand
if lv (critpos \ttgzero) \tttleqA{A} a then vl(critpos \ttgzero) = vl(critpos \ttgone) else True
\end{ttbox}
The indistinguishability relation over infrastructure graphs is lifted to states, that is, the infrastructure
type.
For states $\ttszero$ and $\ttsone$ their indistinguishability can be written as
$\ttszero$ {\ttsim{\texttt{a}}} $\ttsone$.
\begin{ttbox}
{\bf{definition}} indistinguishability ("(\_ \ttsim{(\_)} \_)")
{\bf{where}} \ttszero {\ttsim{\texttt{a}}} \ttsone = indistinguishability\_g (graphI s0) a (graphI s1)
\end{ttbox}  

NI lifts the notion of indistinguishability of states to those of execution sequences of states:
if two states $s_0, s_1$ are indistinguishable, that is, $s_0 \sim_\texttt{a} s_1$, then
for any next state $s_0'$ of $s_0$, that is $s_0 \to s_0'$, 
there exists a state $s_1'$ reachable from $s_1$, that is, with $s_1 \to^* s_1'$,
such that $s_0' \sim_\texttt{a} s_1'$. The states $s_0'$ and $s_1'$ are also indistinguishable and thus indistinguishability
is preserved by the state transition relation.

The property of NI is proved as a theorem for the flightradar system.
The indistinguishable states \texttt{\ttszero} and \texttt{\ttsone} are assumed to be reachable states.
The initial conditions are proved to be invariant. Thus, they hold for \texttt{\ttszero} and \texttt{\ttsone}.
\begin{ttbox}
{\bf{theorem}} noninterference:
 I \ttto{*} \ttszero \ttImp I \ttto{*} \ttsone \ttImp
 clrnc(gra(graphI I)) = lv(gra(graphI I)) \ttImp
 lv(gra(graphI I)) = lv(planes(graphI I)) \ttImp
 lv(planes(graphI I)) = lv(routes(graphI I)) \ttImp
 lv(routes(graphI I)) \tttleA{A} lv(critpos(graphI I) \ttImp
 lv(critpos(graphI I)) = lv(critloc(graphI I)) \ttImp
 \ttszero \ttsim{\texttt{a}} \ttsone \ttImp \ttszero \ttto \ttszero' \ttImp \ttexists \ttsone'. \ttsone \ttto{*} \ttsone' \ttand \ttszero' \ttsim{\texttt{a}} \ttsone'
\end{ttbox}

\section{Shadow definition and example}
\label{sec:shaiiif}
The shadow is inspired by Morgan \cite{mor:09} 
representing all values that a program variable might ``possibly'' have. 
Morgan used an epistemic logic, where the ``possibly'' operator is as usual given as
$\neg(K \neg p)$. In a simple programming language, this allows to express the uncertainty of an
attacker about variable assignments as $\neg(K \neg (x = v))$; 
the set of possible values $v$ of variable $x$ is the shadow.
A refinement is secure if it does not shrink the shadow -- the uncertainty of the attacker must remain
at least the same.
Unlike Morgan, we do not focus on simple programming languages. Our ``language'' are infrastructures, actors,
policies and state transition. Our semantics is the state transition relation describing the infrastructure 
system behaviour. Hence, it is given by the inductive definition in IIIf defining the state transition relation.

The intuition of the shadow, however, is similar to Morgan, the ``ignorance'' of the attacker. That is, the shadow of
a security critical state component contains all the possible values this state component may possibly have
during state transition execution. The shadow thus corresponds to the type of a state component where we mean the
interpretation of a type as the set of all the values it entails. In Isabelle, and hence in IIIf, the set of elements
of a type $\alpha$ can be expressed as \texttt{UNIV :: \ttalpha}.



\subsection{Shadow definition}
We call the attacker \texttt{h} and and their access level \texttt{a = clrnc(h)}.
To implant the shadow with the right security controls, we define information flow dependencies between
the components of the infrastructure state.
\begin{definition}[IF $c_0 \to C$]\label{def:ifctc}
  We say state component $c_0$ has an information flow to the set of components $C$,
  written IF $c_0 \to C$,
  if for some state $s_0$, there is a state $s_1$ with $c_0(s_0) \neq c_0(s_1)$
  while $\forall c \in C.\ c(s_0) = c(s_1)$ and $\exists s_0'.\ s_0 \ttrelI s_0'$, such that for all
  $s_1'$ with $s_1 \ttrelstar{i} s_1'$ we have $\exists c \in C.\ c(s_0') \neq c(s_1')$.
\end{definition}

A shadow is introduced for all security critical state components, that is, all components above
the attacker's clearance level \texttt{a} that have an information flow IF $c_0 \to C$ to components
visible to \texttt{a}.
\begin{definition}[Shadow]\label{def:shadow}
\begin{enumerate}
\item For every component $c_i$ whose level is above the attacker's clearance, i.e \texttt{a}$\tttleA{A}$ \texttt{lv($c_i$)},
  we add an additional shadow component \texttt{shadow\_}$c_i$ to the infrastructure datatype, that is, the state.
\item As an initial value, a shadow component of non-function type $\tau$ is assigned the full universe of its type,
  that is, the set \texttt{UNIV :: \tttau}. For components of function type, the initial value is the function assigning
  the universe as the output for each input. For example, a component $c$ of type \texttt{\tttau \ttfun bool} 
  has initial value \texttt{\ttlam x :: \tttau. \{True, False\}}.
\item Each shadow component \texttt{shadow\_}$c_i$ has a ``shrink-shadow'' clause that reduces the universe initially
  assigned to it. This reduction encodes the uncertainty lost in a step.
  For example, the same component $c$ of type \texttt{\tttau \ttfun bool} with initial value
  \texttt{\ttlam x :: \tttau. \{True, False\}} may be reduced for input \texttt{\ttxzero}, say, by updating 
  the function to \texttt{(\ttlam x :: \tttau. \{True, False\})(\ttxzero := \{True\})}
\item Given the attacker level \texttt{a} the shadow component \texttt{shadow\_}$c_0$
  with protection level \texttt{lv}($c_0$) is guarded by an if-then-else statement shrinking the shadow if \texttt{a}
  is above any component dependent on $c_0$, that is, 
    above the \textbf{greatest lower bound}\footnote{Since
    we use the provided \texttt{IFC} theory, we may assume the type \texttt{ac} of labels to be a lattice. Thus
    greatest lower bounds (glb) and least upper bounds (lub) exist for all sets of security labels.}
    of all levels of state components $c$ dependent on $c_0$ and visible to the attacker, i.e. $\texttt{glb}(\texttt{lv'}(\text{IF}\ c_0\to \{ c.\ \texttt{lv}(c) \tttleqA{A} \texttt{a}\}))$.
\end{enumerate}  
  Summarizing, the shadow components are of the form
\begin{ttbox}
  shadow\_\ttanyany{c}{0} = if \textbf{glb}(\texttt{lv'}(IF \ttanyany{c}{0} \ttto{} \{\ttany{c}. lv(c) \tttleqA{A} a \})) \tttleqA{A} a  
             then "\textit{shrink\_shadow}" else shadow\_\ttanyany{c}{0}
  \end{ttbox}  
\end{definition}
The shadow \texttt{must not} shrink during system execution if the security critical fields should be kept confidential.
This can be expressed as a ``shadow'' invariant, called \texttt{\ttIFC{\texttt{a}}}, of the state transition relation.
\begin{ttbox}
{\bf definition} \ttIFC{\texttt{a}} = \ttforall I s s'. I \ttrelI s \ttImp  a \tttleA{A} lv (critloc G) \ttImp
    s \ttrelI s' \ttImp shadow\_critloc s = shadow\_critloc s'
\end{ttbox}  
In a naive implementation of the flightradar without \texttt{critpos} and hiding (see \cite{kam:24a}),
this invariant does not hold, even though \texttt{a \tttleA{A} lv (critloc G)}. However, as we will see next, in
the secured version it does not shrink.

\subsection{Shadow of the secured example does not shrink}
\label{sec:secexshad}
%
Since the \texttt{igraph} component \texttt{curpos} stores the hidden real location in the case of
an airplane circumventing a critical location, this component is directly coupled with \texttt{critloc}.
Both are identified by the global policy as security critical.
So, for the example, we define a shadow only for \texttt{critloc} and \texttt{curpos}
as security critical components to be hidden from the attacker \texttt{a}.
The shadow is implemented into the graph by adding \texttt{shadow} components for the two security critical
variables \texttt{critloc} and \texttt{critpos} which are sets of all possible values they could take.
\begin{ttbox}
 {\bf shadow\_critloc} :: (location \tttimes location) \ttfun bool set
 {\bf shadow\_critpos} :: identity \ttfun (location\tttimes{location})set
\end{ttbox}  
These shadow components are assigned initially to the values corresponding to those sets of all possible
values
that these two state components may take on during the execution of the state transition.
In the initial state \texttt{I} these values are defined by the universes of the types
(Definition \ref{def:shadow}(2)).
\begin{ttbox}
{\bf shadow\_critloc} I = \ttlam (l,l'). \{True, False\}
{\bf shadow\_critpos} I = \ttlam f.\{x::(location\tttimes{location}). True\}
\end{ttbox}
The semantics of the shadow is added into the inductive definition of the state transition relation.
The rules for \texttt{move} and \texttt{move\_crit} are adapted by integrating guards for both
shadow components (Definition \ref{def:shadow}(4)). Different to the previous example,
the \texttt{glb} term is over their own levels because there are no implicit flows to lower classes
(Definition \ref{def:ifctc}).

\begin{ttbox}
{\bf{move}}: G = graphI I \ttImp (l, l') \ttin  gra G \ttImp h \ttin actors\_graph G \ttImp
  clrnc h = a \ttImp f \ttin planes G (l, l') \ttImp (n,n') = hd(routes G f) \ttImp
  \ttneg critloc G (n, n') \ttImp
  I' = Infrastructure
  (Lgraph
   (gra G)
   (lv (planes G), vl(planes G)((l,l'):= vl(planes G) (l,l')-\{f\})
                               ((n,n'):= vl(planes G)(n,n') \ttcup \{f\}))
   (lv (routes G), vl (routes G)(f := tl(vl(routes G) f)))
   (critloc G)
   (lv (critpos G), (vl (critpos G))(f := None))
   (if (glb\{lv(critloc G), lv(critpos G)\})\tttleqA{A} a
    then (shadow\_critloc G)((l,l'):=\{False\}) else (shadow\_critloc G))
   (if (glb\{lv(critloc G), lv(critpos G)\})\tttleqA{A} a
    then (shadow\_critpos G)(f:= \{(n,n')\}) else (shadow\_critpos G)))
\ttImp I \ttrelI I' 
\end{ttbox}
The difference between the shadow definition now as compared to the insecure example above is the additional
component \texttt{shadow\_critpos}.
\begin{ttbox}
{\bf{move\_crit}}: G = graphI I \ttImp (l, l') \ttin  gra G \ttImp h \ttin actors\_graph G \ttImp
 clrnc h = a \ttImp f \ttin planes G (l, l') \ttImp (n,n') = hd(routes G f) \ttImp
 critloc G (n, n') \ttImp
 I' = Infrastructure
  (Lgraph
   (gra G)
   (lv(planes G), vl(planes G)((l,l'):= vl(planes G)(l,l')- \{f\})
                    ((circumvent(l,l')(n,n')):= vl(planes G)(n,n')\ttcup\{f\}))
   (lv(routes G), vl(routes G)(f := tl(vl(routes G) f)))
   (critloc G)
   (lv (critpos G), vl(critpos G)(f \ttmaplet (circumvent (l,l')(n,n'))))
   (if (glb\{lv(critloc G), lv(critpos G)\})\tttleqA{A} a
    then (shadow\_critloc G)(l,l'):=\{True\} else (shadow\_critloc G))
   (if (glb\{lv(critloc G), lv(critpos G)\}))\tttleqA{A} a 
    then (shadow\_critpos G)(f:= \{ circumvent(l,l')(n,n') \})
    else (shadow\_critpos G)))
 \ttImp I \ttrelI I' 
\end{ttbox}
We can now prove 
that the shadow invariant \texttt{\ttIFC{\texttt{a}}} holds.
\begin{ttbox}
{\bf{theorem}} \ttIFC{\texttt{a}}: \ttforall I s s'. I \ttrelI s \ttImp s \ttrelI s' \ttImp
  a \tttleA{A} lv(critloc (graphI s))) \ttImp a \tttleA{A} lv(critpos (graphI s))) \ttImp
shadow\_critloc s = shadow\_critloc s' \ttand shadow\_critpos s = shadow\_critpos s'  
\end{ttbox}

\subsection{IFC and NI}
\label{sec:IFCNI}
The shadow definition allows us to characterize that an implicit illicit information flow exists by
defining the shadow according to Definition \ref{def:shadow} which uses Definition \ref{def:ifctc} and
then dis-proving the shadow invariant \texttt{\ttIFC{\texttt{a}}}.
This means that if the shadow invariant holds, there are no illicit information flows.
We can thus expect \texttt{\ttIFC{\texttt{a}}} and NI to be equivalent.
We will elaborate this by proving (on paper) an equivalence theorem in this section and later
using it for secure refinement.

\subsubsection{Negation of NI}
\label{eq:nNI}
We expect that if NI does not hold, there are information flows.
NI is defined with respect to an attacker viewpoint \texttt{a} as \ttNI{\texttt{a}}:
\[ \ttNI{\texttt{a}}  \ttequiv  \forall s_0 s_1 s_1'.\ (s_0 \ttsim{\texttt{a}} s_1 \wedge s_0 \to s_0') \longrightarrow 
                \exists s_1'.\ s_1 \ttrelstar{i} s_1' \wedge s_0' \ttsim{\texttt{a}} s_1' \]
The negation of this formula is 
\[
\neg \ttNI{\texttt{a}}  \equiv  \exists s_0 s_1 s_1'.\ s_0 \ttsim{\texttt{a}} s_1 \wedge s_0 \to s_0' \wedge
                \forall s_1'.\ s_1 \ttrelstar{i} s_1' \longrightarrow s_0' \not\sim_\texttt{a} s_1' \hspace{1cm}(4.4)
\]
\subsubsection{Equivalence of $\neg \text{NI}_\texttt{a}$ and IF $c_0 \to C$}
The last line of the previous section (\ref{eq:nNI}) shows that $\neg \text{NI}_\texttt{a}$ is  almost identical
to IF $c_0 \to C$ (see Definition \ref{def:ifctc}).
We use this resemblance to show the equivalence relation between $\neg \text{NI}_\texttt{a}$ and IF $c_0 \to C$.
\begin{theorem}\label{thm:nNIeqIF}
  \[
  \neg \text{NI}_\texttt{a} \equiv \exists c_0.\ \texttt{a} \tttleA{A} \texttt{lv}(c_0) \wedge \text{IF}\ c_0 \to \{ c. \ \texttt{lv}(c) \tttleqA{A} \texttt{a}\} \neq \varnothing \]
\end{theorem}
\begin{proof}
  (``$\Leftarrow$''): With Definition \ref{def:ifctc}, we have $c_0$ with $\texttt{a} \tttleA{A} \texttt{lv}(c_0)$
  and $s_0, s_1, s_0'$ such that $c_0(s_0) \neq c_0(s_1)$, $s_0 \ttrelI s_1$,
  $\forall c \in \{c.\ \texttt{lv}(c) \tttleqA{A} \texttt{a}\}.\ c(s_0) = c(s_1)$ and
  $\forall s_1'.\ s_1 \ttrelstar{i} s_1'$ there is a
  $c \in \{ c.\ \texttt{lv}(c) \tttleqA{A} \texttt{a}\}.\ c(s_0')\ttneq c(s_1')$. That is, $s_0' \not\sim_\texttt{a} s_1'$.

  (``$\Rightarrow$''): From $\neg \text{NI}_\texttt{a}$, we have with (4.4) $s_0, s_1, s_1'$ with $s_0 \sim_\texttt{a} s_1$,
  $s_0 \ttrelI s_0'$ and $\forall s_1'. s_1 \ttrelstar{i} s_1' \ttimp s_0' \not\sim_\texttt{a} s_1'$ (1).
  We first conclude that $s_0 \neq s_1$ because otherwise we could set $s_1'$ to $s_0'$ thus
  $s_1 \ttrelI s_1' = s_0'$ and $s_0' \not\sim_\texttt{a} s_1' = s_0'$ contradicting reflexivity of $\sim_\texttt{a}$.
  Thus, there is a state component $c_0$ with $c_0(s_0) \neq c_0(s_1)$ (because $s_0 \sim_\texttt{a} s_1$ and $s_0 \neq s_1$)
  and $\texttt{a} \tttleA{A} \texttt{lv}(c_0)$.
  Let $s_1'$ be some state such that $s_1 \ttrelstar{i} s_1'$. Using $(1)$, then we have $s_0' \not\sim_\texttt{a} s_1'$.
  That is, $\exists c.\ \texttt{lv}(c) \tttleA{A} \texttt{a}$ and $c(s_0') \neq c(s_1')$. That is with Definition
  \ref{def:ifctc}, we have $\text{IF}\ c_0 \to \{ c. \ \texttt{lv}(c) \tttleqA{A} \texttt{a}\} \neq \varnothing$. \hfill{$\Box$}
\end{proof}  

\subsubsection{$\text{IFC}_{\texttt{a}}, \text{NI}_\texttt{a}$ and IF $c_0 \to C$}
How does the equivalence between $\text{NI}_\texttt{a}$ and the Definition of IF $c_0 \to C$ relate
to the shadow? Definition \ref{def:shadow} gives us that for all $c_0$  with $\texttt{a} \tttleA{A} \texttt{lv}(c_0)$
a shadow component is added and guarded by an if-then-else statement that triggers a "\textit{shrink-shadow}" function
if there is 
some $c \in C$ 
visible to \texttt{a} that receives an information flow.
Thus, the shadow invariant $\ttIFC{\texttt{a}}$ holds iff 
  $\neg(\exists c_0.\ \texttt{a} \tttleA{A} \texttt{lv}(c_0) \wedge \text{IF} c_0 \to \{c.\ \texttt{lv}(c) \tttleqA{A} \texttt{a}\} \neq \varnothing)$.
Because of the equivalence we proved in Theorem \ref{thm:nNIeqIF}, we thus can follow immediately that
the shadow invariant is equivalent to NI
\[ \ttIFC{\texttt{a}} \equiv \text{NI}_\texttt{a}\,. \]

\section{Secure Refinement in IIIf}
\label{sec:secref}
\subsection{Definition of Refinement in IIIf}

Refinement of infrastructures within the IIIf framework is a
relation on Kripke structures that is parameterized by
a polymorphic function that maps the refined type to the abstract type.
This relation \texttt{refinement} is typed as a relation over triples --
a function from a threefold Cartesian product to \texttt{bool}, the 
type containing constants \texttt{True} and \texttt{False} only.  
The type variables $\sigma$ and $\sigma'$ in the type constructor 
\texttt{Kripke} represent the abstract state type and the concrete state type.
Consequently, the middle element of the triples selected by the relation 
\texttt{refinement} is a function of type $\sigma' \Rightarrow \sigma$ 
mapping elements of the refined state to the abstract state.
The expression in quotation marks after the type is the syntax definition
which allows to write a refinement of an abstract Kripke structure \texttt{K} to
a concrete one \texttt{K'} over a type map $\varepsilon$ as \texttt{K \ttmeref K'}.
\begin{ttbox}
refinement :: (\ttsigma kripke \tttimes (\ttsigma' \ttfun \ttsigma) \tttimes \ttsigma' kripke) \ttfun bool ("_ \ttmref{(\_)} _")
 K \ttmeref K' \ttequiv \ttforall s' \ttin states K'. \ttforall s \ttin init K'.
              s \ttrelstar{\sigma'} s' \ttimp \ttecal(s) \ttin init K \ttand \ttecal(s) \ttrelstar{\sigma} \ttecal(s')
\end{ttbox}
The definition of \texttt{K \ttmeref\, K'} states that for any state $s'$ 
of the refined Kripke structure that can be reached by the state transition
in zero or more steps from an initial state $s$ of the refined Kripke 
structure, the mapping ${\mathcal E}$ from the refined to the abstract 
model's state must preserve this reachability, i.e., the image of
$s$ must also be an initial state and from there the image of $s'$
under ${\mathcal E}$ must be reached with $0$ or $n$ steps.

A prominent consequence of the definition of refinement 
is that of property preservation. It shows that refinement preserves the
CTL property of ${\sf EF} s$ which means that a reachability property true in the
refined  model \texttt{K'} 
is already true in the abstract model.
A state set $s'$ represents a property 
in the predicate transformer view of properties as sets of states. That is, properties are elements
of the power set \texttt{Pow(states K')} of the state of Kripke structure \texttt{K'}. 
The additional condition on initial states ensures that
the initial states of the abstract systems must also be represented by initial states in the
concrete system in order to guarantee preservation.
\begin{ttbox}
{\bf{theorem}} prop_pres: K \ttmeref K'  \ttImp init K \ttsubseteq \ttecal\ttimg(init K') \ttImp
\ttforall s' \ttin Pow(states K'). K' \ttvdash {\sf EF} s' \ttimp K \ttvdash {\sf EF} (\ttecal\ttimg(s'))
\end{ttbox}
The definition of refinement by Kripke structure refinement entails property preservation
and proves it as a theorem in Isabelle once for all -- a so-called meta-theorem. So, it can be
used for any concrete refinement later by instantiating the meta-theorem.

Given a global security property \texttt{global\_policy x} in a specification, the refinement
process of the IIIf (the RR-cycle \cite{kam:20a}) interleaves attack tree analysis with refinement
proofs until no more attack paths are found on this security property.
Then, the RR-cycle may achieve a theorem that represents that no more attacks are found
on this global security property.
\begin{ttbox}
  K \ttvdash \ttneg {\sf EF} \{x. \ttneg global\_policy x\}
\end{ttbox}
This theorem can be more directly expressed as ``the global policy holds everywhere''
since the \texttt{{\sf AG} f} statement corresponds to \texttt{\ttneg {{\sf EF} \ttneg f}},
a meta-theorem available in IIIf.
\begin{ttbox}
 K \ttvdash {\sf AG} \{x. global\_policy x \}
\end{ttbox}
This looks very satisfactory: we can prove that globally our requirements hold. There are, however, two 
questions related to this statement:
\begin{itemize}
\item If we arrive at the proof of {\sf AG} \texttt{\{x. global\_policy x\}}, does this mean that
  we can rely on it being true for all further refinements?
\item Does the preservation of properties by IIIf's refinement also include security properties like
  NI?
\end{itemize}  
For the second question, McLean already showed in his seminal paper \cite{mcl:94} that information flow
security properties, like NI, are properties over trace sets and thus in a different class to safety properties.
Thus, they do not adhere to the notion of properties as state sets. Consequently, the property preservation
theorem (shown above) does not apply to NI. While this has tremendous effects on model checking NI, it does not
mean that NI \textit{cannot} be preserved by refinement. We will investigate the conditions under 
which preservation of NI can still be guaranteed in the section after the next.

\subsection{Security does not refine}
\label{sec:ninoref}
As has been discussed at the end of the previous section, security unfortunately does not generally refine.
As an illustration example, we consider the notion of NI as defined and demonstrated on the
flight\-radar system in Section \ref{sec:ni}.
The confidential data in the \texttt{critloc} component was 
hidden by recording the circumvention of airplanes only in the \texttt{critpos} component.
Now, the refinement adds a component that is visible to lower classes and also visible to the attacker.
This component shows the current speed of airplanes. For simplicity, we can compute the speed based on the
grid-like positions-map and calculate the speed per move as $(\Delta x + \Delta y)$ for the horizontal $x$
and vertical $y$ components' absolute differences $\Delta$ between current and next position. Let us further
assume that the speed is calculated on the real speed, that is, the actual position, using the hidden
\texttt{contrapos} rather than the publicly proclaimed one.
In these circumstances, there will be irregularities between the speed calculated based on the ``publicly'' visible
positions and the speed calculated by the real positions (whenever the calculation is based on the hidden position
in \texttt{curpos}).

In this example, the public output and its irregularities are dependent on the hidden security critical component
\texttt{critpos}. Therefore, there is an implicit information flow from a higher classified component to a lower
one. It is a refinement but it violates the security property of NI that had been proved for the abstract specification.

\subsection{IFC allows characterizing secure refinement}
\label{sec:secrefifc}
After these preparations, we are able to express another result: 
the shadow invariant IFC allows to characterize when a refinement preserves security.
This is expressed by the following theorem stating that if the refinement map \texttt{\ttecal}
restricted to the fields of the concrete shadow is injective and is 
compatible with the shadow, that is, it is {\it monomorphic}, then the shadow invariant is preserved from
abstract to the concrete system.
\begin{theorem}[Security preservation]\label{thm:IFCpreservation}
 Let \texttt{\ttref{\ttecal}} be a refinement from an abstract system \texttt{A} to a concrete system \texttt{C}.
 Let  the refinement map \texttt{\ttecal} restricted to
 the shadow of the concrete system \ttrestr{\ttecal}{\ttshadow{C}} be monomporphic. Then 
 the shadow invariant of the abstract system \texttt{\ttIFC{A}}  implies that of the concrete system \texttt{\ttIFC{C}}
 -- both with respect to the same attacker level \texttt{a}.
\begin{ttbox}
\ttforall s. \ttrestr{\ttecal}{\ttshadow{C}}(\ttshadow{C} s) = \ttshadow{A}(\ttecal s) \ttImp
A \ttmeref C \ttImp \texttt{inj}\ttrestr{\ttecal}{\ttshadow{C}} \ttImp \ttIFC{\texttt{A}} \ttImp \ttIFC{\texttt{C}}  
\end{ttbox} 
\end{theorem}
\begin{proof} We need to show
\[
\forall s, s' \in \texttt{C}.\ s \ttrel{\texttt{C}} s' \Rightarrow \texttt{\ttshadow{C}}(s) = \texttt{\ttshadow{C}}(s')
\]
Let  $s, s' \in$ \texttt{C} with $s \ttrel{\texttt{C}} s'$. Because of \texttt{A \ttmeref C}, we have
\[
\ttecal(s) \ttrelstar{\texttt{A}} \ttecal(s')\,. \hspace{4cm}(1)\hspace{-2.5cm}
\]
Because of the first premise, we have
   \[
   \left(
   \begin{array}{ccc}
     \texttt{\ttshadow{A}}(\ttecal(s)) & = & \ttrestr{\ttecal}{\ttshadow{C}}(\ttshadow{C}(s)) \\
     \texttt{\ttshadow{A}}(\ttecal(s')) & = & \ttrestr{\ttecal}{\ttshadow{C}}(\ttshadow{C}(s')) \\
   \end{array}
   \right) \hspace{1.5cm}(2)
   \]
   Since we have \texttt{\ttIFC{\texttt{A}}} and $(1)$, we also have
   \[
   \texttt{\ttshadow{A}}(\ttecal(s)) = \texttt{\ttshadow{A}}(\ttecal(s'))
   \]
extending \texttt{\ttIFC{\texttt{A}}} inductively to $\ttrelstar{A}$ and applying it to $(1)$.
Thus, with $(2)$, we have that
   \[
    \ttrestr{\ttecal}{\ttshadow{C}}(\ttshadow{C}(s)) = \ttrestr{\ttecal}{\ttshadow{C}}(\ttshadow{C}(s'))
    \]
Finally, because of the injectivity of \texttt{\ttrestr{\ttecal}{\ttshadow{C}}} we get
  \[
    \texttt{\ttshadow{C}}(s) = \texttt{\ttshadow{C}}(s')
  \]
as we needed to show. \hfill{$\Box$}   
\end{proof}  

With the equivalence shown in Section \ref{sec:IFCNI}, we can immediately conclude the following corollary.
\begin{corollary}\label{NIeqIFC}
  With the same provisos as in Theorem \ref{thm:IFCpreservation}, we have
  \[ \text{NI}_{\texttt{A}} \Rightarrow \text{NI}_{\texttt{C}}\,.\]
\end{corollary}

\subsection{Discussion}
\label{sec:sndex}
A couple of illustrative examples, provided here informally, shall add some intuition to the
application possibilities of Theorem \ref{thm:IFCpreservation} and its limitations. Both examples are
refinements of the last refined specification in Section \ref{sec:secexshad}.

As a negative example consider first the case where the abstract datatype of the
infrastructure is extended by an additional security critical component \texttt{critact} marking also
flights as critical (for example, depending on passengers on board). The refined semantics
is now that a critical location and a critical flight are necessary for the flight to be
circumvented. That is, \texttt{move\_crit} is conditional on \texttt{critpos (n,n')\ttand critact f}
whereas \texttt{move} on \texttt{\ttneg critpos(n,n')\ttor \ttneg critact f}. As before the
circumvention is only registered in the \texttt{curpos} component. There are no illicit implicit information
flows; the concrete refined specification includes a shadow component for \texttt{critact} that checks
\texttt{(lv(critact G))\tttleqA{A} a}; the concrete shadow has \texttt{\ttIFC{\texttt{C}}}.
We also have \texttt{\ttIFC{\texttt{A}}} as proviso and \texttt{A \ttmeref C } (Section \ref{sec:secexshad}).
However, since all functions in HOL are total, \texttt{\ttrestr{\ttecal}{\ttshadow{C}}} would need to
map the new component to some existing one in the abstract and thus cannot be injective.
Thus we cannot apply Theorem \ref{thm:IFCpreservation} to deduce \texttt{\ttIFC{\texttt{C}}} (even though it is
true).

The second example illustrates a positive case: consider again the refinement from the same already secure example
we have used in \ref{sec:ninoref}.
In this example an additional new component \texttt{speed} is added that is visible to the attacker and is calculated on
the $\Delta$ between real positions (see Section \ref{sec:ninoref}). 
In this second example, the public output and its irregularities are dependent on the hidden security critical component
\texttt{critpos}. Therefore, the component \texttt{\ttshadow{C}} will in the fields \texttt{shadow\_critloc} and
\texttt{shadow\_critpos} shrink in each step (as also state transitions steps with no irregularities contain the
information that this position is not critical). Hence, the premises of Theorem \ref{thm:IFCpreservation}
are not fulfilled, that is, the compatibility condition
$\ttforall s. \ttrestr{\ttecal}{\ttshadow{C}}(\ttshadow{C}\, s) = \ttshadow{A}(\ttecal\, s)$ is not true. Therefore,
we cannot apply Theorem \ref{thm:IFCpreservation}; the refinement is not secure.

Theorem \ref{thm:IFCpreservation} detects attacks where the refinement adds a new component
\texttt{\ttanyany{c}{1}} to the concrete state 
dependent on a confidential component \texttt{\ttanyany{c}{0}} already existing in the abstract state. 
It thus addresses examples like the ones described in Section \ref{sec:ninoref} that add an additional component for
speed (or the decoy QKD case study briefly mentioned \cite{knpw:25}). 
If the level of the new component is visible to the attacker, that is, \texttt{lv(\ttanyany{c}{0}) \tttleA{A} a}
then the shadow of the concrete state would not be compatible with the one of the abstract failing the theorem's
precondition 
\texttt{\ttrestr{\ttecal}{\ttshadow{C}}(\ttshadow{C} s) = \ttshadow{A}(\ttecal s)} because in the abstract
the shadow invariant \texttt{\ttIFC{\texttt{A}}} holds.
If, on the other side, the new component \texttt{\ttanyany{c}{1}} is confidential, that is, above the attacker's viewpoint
\texttt{a \tttleqA{A} lv(\ttanyany{c}{0})}, it might potentially introduce new information flows. This case is excluded
by the precondition that \ttrestr{\ttecal}{\ttshadow{C}} must be injective.

However, in the other cases where the refinement introduces only new components \texttt{\ttanyany{c}{0}}
at a level visible to the attacker, i.e. \texttt{a \tttleqA{A} lv(\ttanyany{c}{0})}, and those
components do not depend on confidential components, Theorem \ref{thm:IFCpreservation} grants that
\texttt{\ttIFC{\texttt{C}}} holds -- provided the refinement map \texttt{\ttecal} is compatible on the shadows.

\section{Conclusions and Related Work}
\label{sec:concl}
In this paper, we have re-considered the application of the IIIf on flightradar systems \cite{kam:24a} and their extension
to information flow control.
We investigated whether a shadow definition could be applied to establish a security preserving refinement for IIIf.
Inspired by the concept of ``ignorance preservation'' presented by Morgan
\cite{mor:09} we have defined a shadow construct for IIIf validating it on the example.
We have investigated how the ignorance preservation can be expressed in the IIIf by an
invariant \texttt{\ttIFC{\texttt{a}}}. Furthermore, we have proved in this paper that the notion of
\texttt{\ttIFC{\texttt{a}}} is equivalent to $\text{NI}_{\texttt{a}}$.
We have then revisited the notion of refinement in IIIf \cite{kam:19a,kam:20a}. Using the shadow and
\texttt{\ttIFC{\texttt{a}}}, we could now characterize the refinement of information flow security 
establishing Theorem \ref{thm:IFCpreservation} that shows that under a compatibility condition on the shadow
with the refinement map \texttt{\ttecal}, security is preserved from the abstract to the concrete (refined) system.

Morgan \cite{mor:09} uses an epistemic logic for the definition of the shadow. This logic is used to express ignorance
as ${\sf P_a}(x = C)  = \neg {\sf K_a} \neg(x = C)$ (``possibly'' {\sf P}) - a common construct in epistemic
logic where knowledge is expressed as ${\sf K_a} p$ for ``agent {\sf a} knows $p$''.
Morgan introduces it as an add on to the refinement calculus because in their framework it is crucial to have the
logical connectives to express uncertainty as a post- or precondition. Thereby, it is possible for them to express
which refinements they accept using knowledge and more importantly ignorance (of the attacker that they do not model
explicitly but implicitly).
By contrast, we do not introduce a dedicated logic to express knowledge or ignorance because we have an explicit
model of actors and their security classes and thus their access to data. Knowledge is implicitly represented by
the values of the components that actors can see and also what they can distinguish in the behaviour of the
system.

The concept of shadow has inspired a wide range of follow up research. Recent successful uses of the concept
are an application for resilience against optimization vulnerabilities \cite{dgpw:24}. This work also uses Isabelle
as a basis but focuses on programming languages.
A recent use of the post-to-pre-style reasoning of the shadow is \cite{Chen2025} where
an expectation-based logic for probabilistic programs is used for the analysis of information-leaks
in source-level programs.
Furthermore, \cite{bgn:21} present an epistemic approach to ignorance preservation for refinements.
They take a similar approach to ours in enforcing that ``refinement steps must not induce observer knowledge
that is not already available in the abstract model''. This conjecture is similar to the injectivity
condition, we use in Theorem \ref{thm:IFCpreservation}. In difference to \cite{bgn:21}, we focus on a specific formal model
of infrastructures in the IIIf which directly supports automated reasoning in Isabelle.
Another experimentation with the relationship of distributed GIS systems and machine
learning within the  IIIf leads to an extension by a notion of Federated Learning for Differential
Privacy \cite{kam:24b}. To this end, a distributed version of NI based on probabilities is formalized
vaguely inspired by Volpano and Smith's \cite{vs:99} intuition of how to express probabilistic NI
for distributed systems which differs substantially from the possibilistic notion investigated here.


\begin{thebibliography}{10}

\bibitem{bgn:21}
C.~Baumann, M.~Dam, R.~Guancialey, and H.~Nematiz.
\newblock On compositional information flow aware refinement.
\newblock In {\em IEEE Computer Security Foundations Symposium}, 2021.

\bibitem{Chen2025}
C.~Chen, A.~McIver, and C.~Morgan.
\newblock {\em Source-Level Reasoning for Quantifying Information Leaks},
  pages 98--127.
\newblock Springer Nature Switzerland, Cham, 2025.

\bibitem{suc:16}
CHIST-ERA.
\newblock Success: Secure accessibility for the Internet of Things, 2016.
\newblock http://www.chistera.eu/projects/success.

\bibitem{den:82}
D.~Denning.
\newblock {\em Cryptography and Data Security}.
\newblock Addison-Wesley, 1982.
\newblock reprinted with corrections, Janauary 1983.

\bibitem{de:76}
D.~E. Denning.
\newblock Lattice model of secure information flow.
\newblock {\em Communications of the ACM}, 19(5):236--242, 1976.

\bibitem{dd:77}
D.~E. Denning and P.~J. Denning.
\newblock Certification of programs for secure information flow.
\newblock {\em Communications of the ACM}, 20(7), 1977.

\bibitem{dgpw:24}
B.~Dongol, M.~Griffin, A.~Popescu, and J.~Wright.
\newblock Relative security: Formally modeling and (dis)proving resilience
  against semantic optimization vulnerabilities.
\newblock In {\em IEEE Computer Security Foundations Symposium}, 2024.

\bibitem{fli:24}
  Flightradar24.
  \url{https://www.flightradar24.com}
\newblock Live air traffic.
\newblock Accessed 22.6. 2026.

\bibitem{gm:82}
J.~Goguen and J.~Meseguer.
\newblock Security policies and security models.
\newblock In {\em Symposium on Security and Privacy, SOSP'82}, pages 11--22.
  IEEE Computer Society Press, 1982.

\bibitem{kam:18b}
F.~Kamm\"uller.
\newblock Attack trees in Isabelle.
\newblock In {\em 20th International Conference on Information and
  Communications Security, ICICS2018}, volume 11149 of {\em LNCS}. Springer,
  2018.

\bibitem{kam:19a}
F.~Kamm\"uller.
\newblock Combining secure system design with risk assessment for IoT
  healthcare systems.
\newblock In {\em Workshop on Security, Privacy, and Trust in the IoT,
  SPTIoT’19, colocated with IEEE PerCom}. IEEE, 2019.

\bibitem{kam:23b}
F.~Kamm\"uller.
\newblock Introducing distributed ledger security into system specifications
  with the Isabelle RR-cycle.
\newblock In {\em Computer Security. ESORICS 2023 International Workshops.},
  volume 14399 of {\em LNCS}. Springer, 2023.

\bibitem{kam:24a}
F.~Kamm\"uller.
\newblock Analyzing air-traffic security using GIS-''blur'' with information
  flow control in the IIIf.
\newblock In {\em The 19th International Conference on Availability,
  Reliability and Security Proceeedings, ARES 2024}. ACM, 2024.

\bibitem{kam:26iiif}
F.~Kamm\"uller.
\newblock {IIIf sources -- Isabelle Insider and Infrastructure framework with
  Kripke strutures, CTL, attack trees, security refinement, and applications},
  2026.
\newblock Available at \url{https://gitlab.com/falkam/IIIf}.

\bibitem{knpw:25}
F.~Kamm\"uller, R.~Nagarajan, M.~C. Parker, and C.~White.
\newblock Formalisation and analysis of decoy qkd in the iiif using refinement
  and attack trees.
\newblock In {\em Systems, Men, and Cybernetics, SMC 2025}. IEEE, 2025.

\bibitem{kam:24b}
F.~Kamm\"uller, L.~Piras, B.~Fields, and R.~Nagarajan.
\newblock Formalizing federated learning and differential privacy for GIS
  systems in IIIf.
\newblock In {\em Computer Security -- ESORICS Workshops; SecAssure 2024},
  LNCS. Springer, 2024.
\newblock Proceedings to appear.

\bibitem{kam:20a}
F.~Kammüller.
\newblock A formal development cycle for security engineering in Isabelle,
  2020.
\newblock arxiv preprint, \url{http://arxiv.org/abs/2001.08983}.

\bibitem{mcl:94}
J.~Mclean.
\newblock A general theory of composition for trace sets closed under selective
  interleaving functions.
\newblock In {\em In Proc. IEEE Symposium on Security and Privacy}, pages
  79--93, 1994.

\bibitem{mor:09}
C.~Morgan.
\newblock The shadow knows: Refinement and security in sequential programs.
\newblock {\em Science of Computer Programming}, 74:629--653, 06 2009.

\bibitem{ml:98}
A.~C. Myers and B.~Liskov.
\newblock Complete, safe information flow with decentralized labels.
\newblock In {\em Proceedings of the IEEE Symposium on Security and Privacy}.
  IEEE, 1999.

\bibitem{npw:02}
T.~Nipkow, L.~C. Paulson, and M.~Wenzel.
\newblock {\em Isabelle/HOL -- A Proof Assistant for Higher-Order Logic},
  volume 2283 of {\em LNCS}.
\newblock Springer-Verlag, 2002.

\bibitem{vs:99}
D.~M. Volpano and G.~Smith.
\newblock Probabilistic noninterference in a concurrent language.
\newblock {\em J. Comput. Secur.}, 7(1), 1999.

\end{thebibliography}
\end{document}